\documentclass[12pt]{amsart}

 \usepackage[margin=1in]{geometry}

\newcommand{\private}[1]{}
\usepackage{amssymb, amsmath, amscd, amsthm, color, epsfig,url, graphicx, url}

\usepackage[all]{xy}          
\xyoption{dvips}              

\usepackage{subcaption}
\usepackage{caption}

\makeatletter
\renewcommand\l@subsection{\@tocline{2}{0pt}{2pc}{5pc}{}}
\makeatother

\newcommand{\PolStr}{\operatorname{PolStr}}
\newcommand{\SCom}{\operatorname{SCom}}
\newcommand{\stab}{\operatorname{stab}}
\newcommand{\via}{\operatorname{via}}

\newcommand{\st}{\operatorname{st}}
\newcommand{\wstab}{\operatorname{wstab}}
\newcommand{\type}{\operatorname{type}}

\theoremstyle{plain}
\newtheorem{thm}{Theorem}[section]
\newtheorem{prop}[thm]{Proposition}

\newtheorem{cor}[thm]{Corollary}
\newtheorem{conj}[thm]{Conjecture}

\theoremstyle{definition}
\newtheorem{defin}[thm]{Definition}
\newtheorem{example}[thm]{Example}
\newtheorem{def/ex}[thm]{Definition/Example}

\theoremstyle{remark}
\newtheorem{rem}[thm]{Remark}

\usepackage{tikz}

\makeatletter
\@namedef{subjclassname@2020}{%
  \textup{2020} Mathematics Subject Classification}
\makeatother

\usepackage{graphicx}

\begin{document}
\pagestyle{plain}
\title{Stability of political structures modeled by simplicial complexes under mediation, splitting, and shellability}
\author{Du\v sko Joji\'c}
\address{Faculty of Natural Sciences and Mathematics, University of Banjaluka}
\email{dusko.jojic@pmf.unibl.org}

\author{Franjo \v Sar\v cevi\'c}
\address{Faculty of Civil Engineering, University of Rijeka}
\email{franjo.sarcevic@uniri.hr}
\urladdr{https://sites.google.com/view/franjosarcevic}

\subjclass[2020]{05A99, 91C99, 91F10}
\keywords{political structure; simplicial complex; stability; mediator; splitting; shellable complex; independence complex}

\begin{abstract}
Modeling political structures by simplicial complexes, we investigate whether introducing a mediator into a substructure increases or decreases the stability of the overall structure. We prove theorems that quantify the stability of a political structure when $n$ mediators are introduced, either one by one or simultaneously. We also examine how the stability is affected when a single agent is split into two. In addition, stability is expressed in terms of the $h$-vector, and special attention is given to a class of political structures modeled by shellable simplicial complexes. In the latter context, we analyze weighted political structures and examples of political structures modeled by independence complexes of graphs. This approach provides a rigorous, stepwise analysis of stability under different structural modifications, showing how the combinatorial and topological properties of the simplicial complex govern the structure’s stability.
\end{abstract}

\maketitle


\tableofcontents


\parskip=5pt
\parindent=0cm




\begin{section}{Introduction}
Building on the work presented in \cite{Keiding}, the authors in \cite{MV21} model political structures using simplicial complexes. Modeling political structures in this way is natural, given that these complexes -- specifically, their constituent elements, the simplices, of which they are built by gluing in a prescribed manner -- are well-suited for representing situations involving mutual interactions between objects. Two interacting objects form a 1-simplex, i.e., an edge; three interacting objects form a 2-simplex, i.e., a (full) triangle; four interacting objects form a 3-simplex, i.e., a (full) tetrahedron; and so on.

Such use of simplicial complexes is present in many fields, from topological data analysis \cite{Gunnar, Munch, Otter} and signal processing \cite{Barba, Chad, Schaub} to the mathematical methods in social sciences \cite{Egan, Faridi, Hansen, Iacopo, Martino, Moore}.


In the specific case that is relevant to us here, the authors in \cite{MV21} establish a natural correspondence (see the text following Definition \ref{D:simplcompl}) between the political structure (Definition \ref{D:politstr}) on the one hand and the simplicial complex (Definition \ref{D:simplcompl}) on the other. 


Expressed in the language of category theory, there is an equivalence 
$$\PolStr\simeq \SCom$$
between the category $\PolStr$ of political structures (with \textit{political structure maps} 
as morphisms) and the category $\SCom$ of simplicial complexes (with \textit{simplicial maps} as morphisms).

In this way, problems from political structures are ``translated" into the language of simplicial complexes. This approach enables us to formulate concepts such as the stability of a political structure, merging of political structures, mediation, delegation or compromise in a more mathematically rigorous and operational manner, and to prove as \textit{theorems}, for example, that merging agents from disconnected components of a political structure increases the stability of that structure, that introducing a mediator increases the stability of a political structure, among others. Some facts can be expressed in the language of algebraic topology, using homology groups \cite[Sections 3.5. and 4.]{MV21}, \cite[Sections 2.4. and 3.4.]{Wang}.

In this paper, we investigate the effects of introducing a mediator into a substructure of a political structure, motivated by an open question posed after \cite[Proposition 3.15]{MV21}, as well as the effects of introducing $n$ mediators into a political structure. We also investigate what happens when instead of merging agents we have splitting of an agent - in other words, when a vertex in a simplicial complex is split into two new vertices. Such situations are natural to observe, bearing in mind, for example, the splitting of political parties into new parties. Furthermore, we describe the stability of a political structure modeled by a simplicial complex via the $h$-vector of the complex, with particular attention to the case of shellable complexes, which are considered primarily for their topological and combinatorial properties, but we show that they also possess substantive social-scientific relevance.

In the end, we model political structures as independence complexes of graphs, and illustrate this with several examples, showing how tools from topological combinatorics can be applied to investigate their stability.

This stepwise, combinatorial, and topological approach provides a rigorous framework to quantify stability in political structures modeled by simplicial complexes, with previous works -- supplemented here by new results -- demonstrating the relevance and applicability of this methodology.

\end{section}

\begin{section}{Overview of some basic concepts and results}\label{S:Basics}
The cited literature - especially \cite{Keiding, Brooks, MV21} - provides a detailed introduction to the terminology of political structures, as well as the necessary background on simplicial complexes. Accordingly, here we will limit ourselves to a brief review of key concepts and results that are directly relevant to the problems under consideration, and construct several illustrative examples.

\begin{defin}\label{D:politstr}
A \textit{political structure} $P = (A, \mathcal{C})$ is an ordered pair consisting of a finite set $A$ of \textit{agents} and a collection $\mathcal{C}$ of \textit{viable configurations} -- subsets of $A$ that satisfy two natural conditions: (i) each agent individually forms a viable configuration, and (ii) every subset of a viable configuration is a viable configuration.
\end{defin}

\begin{defin}\label{D:simplcompl}
A finite \textit{(abstract) simplicial complex} $K = (V, \Delta)$ is an ordered pair consisting of a finite set $V$ of \textit{vertices} and a collection $\Delta$ of \textit{simplices} -- subsets of $V$ such that a subset of a simplex is a simplex, where the vertices can be treated as $0$-simplices. 
\end{defin}

Any simplex in the complex is also called a \textit{face} of the complex.

The abstract simplicial complex can be identified with its \textit{geometric realization} without risk of ambiguity; hence, we will simply use the term \textit{simplicial complex}.

A political structure can be modeled by a simplicial complex via a natural correspondence between $P$ and $K$, in which agents are mapped to vertices and viable configurations to simplices.



In this way, a \textit{political structure map} between two political structures, that sends agents to agents and viable configurations to viable configurations, corresponds to a \textit{simplicial map} between the associated simplicial complexes, that sends vertices to vertices and simplices to simplices.

Since the political structure $P$ is modeled by a simplicial complex $K$, we will, from this point on, use the notations $P$ and $K$, as well as the terms \textit{political structure} and \textit{simplicial complex}, interchangeably. That is, when we say that $P$ is a political structure, we implicitly mean the simplicial complex that models it. Conversely, we may refer to the political structure simply as the simplicial complex $P$.

\begin{example}\label{E:structure1}
Figure \ref{F:politstr1} depicts a political structure with six agents $\{v_0,v_1,\ldots,v_5\}$, in which agents $v_0,v_1,v_2$ are mutually compatible, agents $v_2$ and $v_3$, as well as $v_3$ and $v_4$, and $v_2$ and $v_4$ are compatible in pairs, but they cannot all be together in the same coalition, and agent $v_5$ is isolated - it is not compatible with any other agent in the structure.
\end{example}

\begin{example}\label{E:structure2}
Figure \ref{F:politstr2} depicts a political structure with five agents $\{v_0,v_1,\ldots,v_4\}$, where agents $v_1,v_2,v_3,v_4$ are all compatible among themselves (they form a full tetrahedron), but agent $v_0$ is compatible only with $v_1$ - it can form a coalition exclusively with $v_1$ and with no other agent in the structure.
\end{example}

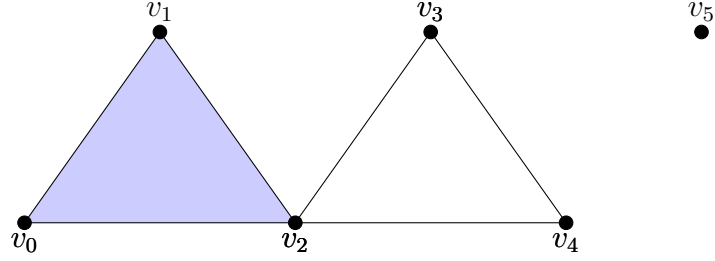
\begin{figure}
\begin{center}
\begin{tikzpicture}[scale=0.9] 
\begin{scope}
\fill[blue!20] (0,0) -- (4,0) -- (2, 2.82);
\end{scope}
\draw[fill] (0,0) circle [radius=0.1]; \draw[fill] (4,0) circle [radius=0.1]; \draw[fill] (2,2.82) circle [radius=0.1]; \draw[fill] (8,0) circle [radius=0.1]; \draw[fill] (6,2.82) circle [radius=0.1]; \draw[fill] (10,2.82) circle [radius=0.1] node [above] {$v_5$};
\draw (0,0) node[below] {$v_0$} -- (4,0) node [below] {$v_2$} -- (2,2.82) node[above] {$v_1$} -- (0,0) node[below] {$v_0$}; 
\draw  (4,0) node [below] {$v_2$} -- (8,0) node [below] {$v_4$}; \draw  (4,0) node [below] {$v_2$} -- (6,2.82) node [above] {$v_3$};  \draw  (8,0) node [below] {$v_4$} -- (6,2.82) node [above] {$v_3$};
\end{tikzpicture}
\end{center}
\caption{An example of a political structure}
\label{F:politstr1}
\end{figure}

\begin{figure}
\begin{center}
\begin{tikzpicture}[scale=0.9] 
\begin{scope}
\fill[blue!10] (0,0) -- (4,0) -- (5.5, 1.41); \fill[blue!20] (0,0) -- (4,0) -- (2, 2.82); \fill[blue!15] (4,0) -- (5.5,1.41) -- (2, 2.82);
\end{scope}
\draw[fill] (0,0) circle [radius=0.1] node [below] {$v_2$}; \draw[fill] (4,0) circle [radius=0.1] node [below] {$v_3$}; \draw[fill] (2,2.82) circle [radius=0.1] node [left] {$v_1$};
\draw (0,0) -- (4,0) -- (2,2.82) -- (0,0); \draw[fill] (5.5,1.41) circle [radius=0.1] node [right] {$v_4$}; \draw (4,0)--(5.5,1.41); \draw (5.5,1.41)--(2,2.82); \draw[dashed] (0,0)--(5.5,1.41); \draw[fill] (4,5.64) circle [radius=0.1] node [above] {$v_0$}; \draw (4,5.64) --(2,2.82);
\end{tikzpicture}
\end{center}
\caption{Another example of a political structure}
\label{F:politstr2}
\end{figure}
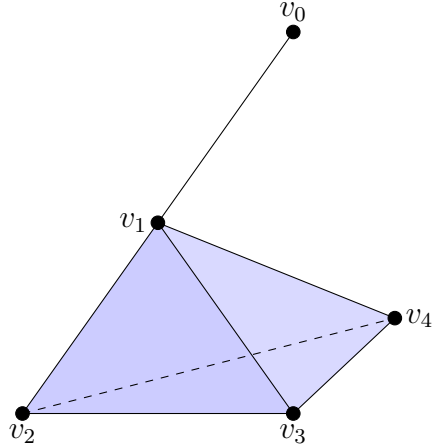

The dimension of $P$ is defined as the maximal dimension of its simplices. For example, the dimension of the structure in Example~\ref{E:structure1} is $2$, while the dimension of the structure in Example~\ref{E:structure2} is $3$.

Given a $d$-dimensional political structure $P$ with $k+1$ agents $\{v_0,v_1,\ldots,v_k\}$, let the $(d+2)$-tuple $f(P)=(f_{-1},f_0,f_1,\ldots,f_d)$ denote the $f$-vector of $P$, i.e., for $i\geq 0$, $f_i$ is the number of $i$-dimensional faces of the simplex, and $f_{-1}=1$. Obviously, $f_0=k+1$.

For a structure $P$, it is defined (\cite[Definition 3.8]{MV21}) a numerical characteristic $0\leq \stab(P)\leq 1$, and for each agent $v_i$ in $P$, it is defined a numerical characteristic $0\leq \via(v_i) \leq 1$.

\begin{defin}
Given a $d$-dimensional political structure $P$ with the set of agents $V=\{v_0,v_1,\ldots,v_k\}$, where $k\geq 1$,
\begin{itemize}
\item the \textit{stability} of $P$ is defined to be
$$\stab(P)=\frac{f_1+\cdots+f_d}{2^{k+1}-k-2},$$ and
\item the \textit{viability} of $v_i$ is defined to be $$\via(v_i)=\frac{1}{2^k-1}\left(\vert \st(v_i)\vert -1\right),$$
where $\st(v_i)$ is the \textit{star} of $v_i$, i.e. the collection of all simplices of $P$ containing $v_i$, and $\vert \cdot \vert$ denotes the cardinality.
\end{itemize}
\end{defin}

For $k=0$, we define by convention $\stab(P)=0$ and $\via(v_0)=0$. In the sequel, we assume that $P$ consists of at least two agents.

If a political structure $P$ consists only of agents (vertices), with no $i$-dimensional simplices (i.e., no compatible subsets) for $i > 0$, then $\stab(P) = 0$. The stability of $P$ is equal to $1$ if $P$ is a simplex -- that is, if all agents form a single coalition.

We note that the intuitive concept of stability of a political structure cannot be fully captured by a single numerical index; however, it is natural in this context to adopt the notion introduced in the works upon which we build. We have addressed this issue in several details in the final section.

\begin{example}
The stability of the structure in Example \ref{E:structure1} is $7/57$, while the stability of the structure in Example \ref{E:structure2} is $6/13$.
\end{example}

\begin{example}
The viabilities of the agents in Example \ref{E:structure1} are: $\via(v_0)=\via(v_1)=3/31, \via(v_2)=5/31, \via(v_3)=\via(v_4)=2/31, \via(v_5)=0$. The viabilities of the agents in Example \ref{E:structure2} are: $\via(v_0)=1/15, \via(v_1)=8/15, \via(v_2)=\via(v_3)=\via(v_4)=7/15$. 
\end{example}

We provide the following example to illustrate and motivate several concepts discussed later in the text. We will refer to it several times.

\begin{example}\label{E:important}
Let $P_{a,b;s}$ be a simplicial complex modelling a political structure, obtained by gluing an $(a-1)$-dimensional simplex $\Delta^{a-1}$ and an $(b-1)$-dimensional simplex $\Delta^{b-1}$ along a common $(s-1)$-dimensional simplex $\Delta^{s-1}$, where $1\leq s < a\leq b$. In short,
$$P_{a,b;s}=\Delta^{a-1}\bigcup_{\Delta^{s-1}}\Delta^{b-1}.$$

Figure \ref{F:politstr-imp} depicts $P_{3,4;2}$. A direct computation yields
$$\stab(P_{a,b;s})=\frac{2^a-a+2^b-b-2^s+s-1}{2^{a+b-s}-a-b+s-1}.$$
Although $P_{a,b;s}$ has a large number of simplices when $a$ or $b$ is large, its stability can be computed explicitly. Moreover, $\stab(P_{a,b;s})$ is an increasing function of $s$ for fixed $a$ and $b$. 

The viability of a vertex in $P_{a,b;s}$ takes just two (for $a=b$) or three (for $a\neq b$) values, namely
$$\frac{2^{a-1}-1}{2^{a+b-s-1}-1}, \frac{2^{b-1}-1}{2^{a+b-s-1}-1}, \frac{2^{a-1}+2^{b-1}-2^{s-1}-1}{2^{a+b-s-1}-1}.$$
\end{example}

\begin{figure}
\begin{center}
\begin{tikzpicture}[scale=0.9] 
\begin{scope}
\fill[blue!10] (0,0) -- (4,0) -- (5.5, 1.41); \fill[blue!20] (0,0) -- (4,0) -- (2, 2.82); \fill[blue!15] (4,0) -- (5.5,1.41) -- (2, 2.82); \fill[red!15] (4.97,5.14)--(2,2.82)--(5.5,1.41);
\end{scope}
\draw[fill] (0,0) circle [radius=0.1]; \draw[fill] (4,0) circle [radius=0.1]; \draw[fill] (2,2.82) circle [radius=0.1];
\draw (0,0) -- (4,0) -- (2,2.82) -- (0,0); \draw[fill] (5.5,1.41) circle [radius=0.1]; \draw (4,0)--(5.5,1.41); \draw[green, line width=3pt] (5.5,1.41)--(2,2.82); \draw[dashed] (0,0)--(5.5,1.41); \draw[fill] (4.97,5.14) circle [radius=0.1]; \draw (4.97,5.14) --(2,2.82); \draw (4.97,5.14) --(5.5,1.41); \draw (4.97,5.14)--(2,2.82)--(5.5,1.41);
\end{tikzpicture}
\end{center}
\caption{The political structure $P_{3,4;2}$}
\label{F:politstr-imp}
\end{figure}

For the $f$-vector $(f_{-1},f_0,f_1,\ldots,f_d)$ of a political structure $P$ with the set of vertices $V(P)=\{v_0,v_1,\ldots,v_k\}$, the $f$-polynomial of $P$ is defined to be
$$f_P(x)=f_{-1}+f_{0}x+f_1 x^2 +\cdots +f_d x^{d+1}.$$
For our purposes here we consider the shifted $f$-polynomial of $P$ to be
$$F_P(x)=f_1 x+f_2 x^2 +\cdots +f_d x^d.$$

Then $$\stab(P)=\frac{F_P(1)}{2^{k+1}-k-2},$$
and $$\sum_{v\in V(P)} \via(v)=\frac{F_P'(1)+F_P(1)}{2^k-1}.$$

The relation for stability is obvious. The relation for the sum of viabilities follows from the following observation: when counting  $\via(v)$, we count all simplices that contain $v$, excluding the vertex $\{v\}$ itself. Thus, every $(i-1)$-dimensional simplex for $i>1$ will be counted $i$ times in the sum $\sum_{v\in V(P)} \via(v)$. Therefore, the numerator in the sum of viabilities equals
$$2f_1+3f_2+4f_3+\cdots + (d+1)f_d=F_P'(1)+F_P(1).$$

\begin{example}
The shifted $f$-polynomial of the structure in Example \ref{E:structure1} is $F_P(x)=6x+x^2$.

The shifted $f$-polynomial of the structure in Example \ref{E:structure2} is $F_P(x)=7x+4x^2+x^3$.

The shifted $f$-polynomial of the structure $P_{a,b;s}$ from Example \ref{E:important} is
$$F_{P_{a,b;s}}(x)=\sum_{i=2}^{b}\left(\binom{a}{i} + \binom{b}{i} - \binom{s}{i}\right)x^{i-1},$$
where, by convention, $\binom{n}{k}=0$ for $k>n$. 
\end{example}

\begin{defin}
A \textit{mediator} in a political structure $P$ is an agent which is compatible with every viable configuration in $P$; in other words, an agent which is compatible with all other agents in $P$ and willing to join any potential coalition.
\end{defin}

Introducing an outside mediator $v$ into a political structure corresponds to taking the \textit{cone} $CP=P*\{v\}$ on $P$, where $*$ denotes the \textit{join} of $P$ and $\{v\}$. See \cite[A3]{MV21} for these and other operations on simplicial complexes.

This represents the situation in which a mediator is introduced to the entire structure, being compatible with all existing agents.

The following result is taken from \cite[Proposition 3.14]{MV21}. For completeness -- and because mediation is one of the main topics discussed here -- we provide a slightly different proof.

\begin{prop}\label{P:mediation1}
Introducing a mediator into a political structure $P$ with $k+1$ vertices increases its stability, i.e. $$\stab(P)\leq \stab(CP),$$ with equality holding if and only if $P$ is a $k$-simplex.
\end{prop}
\begin{proof}
If the number of positive-dimensional simplices of $P$ is $f_1+\cdots+f_d$, then the number of positive-dimensional simplices of $CP$ is $2(f_1+\ldots+f_d)+k+1$. The inequality $\stab(P)\leq \stab(CP)$ is equivalent to 
$$\frac{f_1+\cdots+ f_d}{2^{k+1}-k-2}\leq \frac{2(f_1+\cdots +f_d)+k+1}{2^{k+2}-k-3},$$ which readily reduces to $$f_1+\cdots +f_d \leq 2^{k+1}-k-2.$$ This holds because the maximal number of positive-dimensional simplices in a complex on the vertex set $V=\{v_0, v_1, \ldots, v_k\}$ is $2^{k+1}-k-2$, which counts all non-empty, non-singleton subsets of $V$. This maximum is attained precisely when $P$ is a $k$-simplex.
\end{proof}
\end{section}

\begin{section}{Effects of introducing a mediator into the substructure of a political structure}\label{S:Mediator}
The authors in \cite{MV21} raised a question about what happens when a mediator is introduced into a substructure of a political structure $P$, which corresponds to taking a cone on a subcomplex of the simplicial complex $P$. Politically, this describes cases where some agents are unable to reach agreement on a particular issue, requiring a form of targeted mediation -- where a mediator is introduced specifically between certain agents, with the aim of initiating dialogue within that subgroup. 

For a subcomplex $N\subset P$, let $P_N'$ denote the simplicial complex $P_N'=P\setminus N \coprod CN$ obtained by adding a new vertex $v$ and a cone $CN=N*\{v\}$ on $N$. The authors in \cite[Proposition 3.15]{MV21} provide the following result, which we reproduce here in a slightly reformulated version:

\begin{prop}\label{P:stab}
Suppose $P$ is a political structure with $k+1$ vertices and $N$ is a subcomplex of $P$. Let $G$ denote the number of simplices of $P\setminus N$, i.e. the total number of simplices that are left unaffected. If $G\geq k+1$, then $\stab(P)>\stab(P_N')$.
\end{prop}
\begin{proof}
Let $F$ be the number of simplices of $N$. Then the total number $S$ of simplices of $P$ is $S=F+G$, and the total number of simplices of $P_N'$ is $S'=2F+G+1$. After some straightforward calculations, the inequality $$\frac{S-(k+1)}{2^{k+1}-(k+2)}=\stab(P)>\stab(P_N')=\frac{S'-(k+2)}{2^{k+2}-(k+3)}$$ reduces to $$(k+1)F>(k+1-G)(2^{k+1}-1),$$ which is true for $G\geq k+1$.
\end{proof}

We made a simple observation that identifies a condition under which the stability of the structure increases or decreases, not restricted to the case $G\geq k+1$, thereby providing a more complete answer to the aforementioned question from \cite{MV21}. We state it as an immediate corollary to the previous proposition.

\begin{cor}\label{C:stab}
The inequality $$\stab(P)>\stab(P_N')$$ holds if and only if $$G>\frac{(k+1)(2^{k+1}-S-1)}{2^{k+1}-k-2}.$$
\end{cor}
\begin{proof}
After substituting $F=S-G$, the inequality $$(k+1)F>(k+1-G)(2^{k+1}-1)$$ is equivalent to $$G>\frac{(k+1)(2^{k+1}-S-1)}{2^{k+1}-k-2}.$$
\end{proof}


Therefore, $G$ is a characteristic indicating how benevolent (acceptable) a mediator must be in order to increase stability. Like the main notion $\stab(P)$, it depends only on the total number of vertices and the number of simplices of $P$. In Remark \ref{R:G}, we give a consequence and an interpretation of this result, which actually provide a complete answer to the question from \cite{MV21} that we consider.

The condition stated in Corollary \ref{C:stab} subsumes the one in Proposition \ref{P:stab}, since $$k+1\geq \frac{(k+1)(2^{k+1}-S-1)}{2^{k+1}-k-2}$$ is equivalent to the obvious inequality $$S\geq k+1.$$

\begin{example}\label{E:structures1}\hfill
\begin{itemize}
\item[a)] If $P=\{\{v_0,v_1\},\{v_0\},\{v_1\},\{v_2\},\{v_3\},\{v_4\}\}$ is a political structure with $k+1=5$ agents and $S=6$ simplices, then taking the cone on the subcomplex $N=\{\{v_2\},\{v_3\},\{v_4\}\}$ leaves $G=3$ simplices unaffected. In this case we have $3=G< \frac{(k+1)(2^{k+1}-S-1)}{2^{k+1}-k-2}=\frac{5\cdot 25}{26}$, so the stability increases upon introducing a mediator. Indeed, by direct computation we obtain $1/26=\stab(P)<\stab(P_N')=4/57$. The same is true for $G=4$.
\item[b)]  If $P=\{\{v_0,v_1,v_2\},\{v_0,v_1\},\{v_1,v_2\},\{v_0,v_2\},\{v_1,v_3\},\{v_1,v_4\},\{v_0\},\{v_1\},\{v_2\},\{v_3\},\\\{v_4\}\}$ is a political structure with $k+1=5$ agents and $S=11$ simplices, then taking the cone on the subcomplex $N=\{\{v_0,v_1,v_2\},\{v_0,v_1\},\{v_1,v_2\},\{v_0,v_2\},\{v_0\},\{v_1\},\{v_2\}\}$ leaves $G=4$ simplices unaffected. In this case we have $4=G> \frac{(k+1)(2^{k+1}-S-1)}{2^{k+1}-k-2}=\frac{5\cdot 20}{26}$, so the stability decreases upon introducing a mediator. Indeed, by direct computation we obtain $3/13=\stab(P)>\stab(P_N')=13/57$.
\item[c)] The case in which stability remains unchanged upon introducing a mediator into a proper substructure, that is, when
$$G = \frac{(k+1)(2^{k+1} - S - 1)}{2^{k+1} - k - 2},$$
is, in fact, impossible. Since $S \geq k+1$, it follows that
$2^{k+1} - S - 1 \leq 2^{k+1} - k - 2$,
with equality only when $S = k+1$. Therefore, the expression above can be an integer only if $S = k+1$, in which case $G = k+1$. But this would imply that all simplices are unaffected - that is, no mediator has been introduced at all.
\end{itemize}
\end{example}

For a fixed number of vertices (agents) $k+1$, consider the function $$g(S)=\frac{(k+1)(2^{k+1} - S - 1)}{2^{k+1} - k - 2},$$ where $S\in \{k+1,k+2,\ldots,2^{k+1}-1\}$.
The function $g(S)$ is strictly decreasing. Consequently, as $S$ increases, the number of unaffected simplices $G<g(S)$ for which the introduction of a mediator leads to an increase in system stability decreases. In other words, the more ``developed" the structure -- i.e., the greater the number of simplices or pre-existing compatibilities among agents -- the larger the subcomplex (in terms of cardinality) into which the mediator must be introduced in order to achieve an increase in overall structure stability.

In the special case when $S=2^{k+1}-1$, it follows that $g(S)=0$. Here, a non-strict increase in stability would require $G\leq 0$, that is, $G=0$. Therefore, for the full simplex $\Delta^k$, introducing a mediator leads to a decrease in stability unless the mediator is introduced into the entire structure without excluding any part. Only under this condition does the system maintain its maximum stability value of $1$, both before and after the intervention.

\begin{rem}\label{R:G}
Given a political structure $P$ with $k+1$ vertices, $$\stab(P)\in \left\{0,\frac{1}{2^{k+1}-k-2},\frac{2}{2^{k+1}-k-2},\ldots,1-\frac{1}{2^{k+1}-k-2},1\right\}\subset [0,1].$$ These points are equally spaced within the unit interval $[0,1]$. Let $G$ denote the number of unaffected simplices from Proposition \ref{P:stab} and Corollary \ref{C:stab}. The case $G\geq k+1$ is clear. Let $G\in \{0,1,\ldots,k+1\}$. Then the points $1-\frac{G}{k+1}$ partition the interval $[0,1]$ into $k+1$ congruent subintervals. It is easy to see that the condition $$G>\frac{(k+1)(2^{k+1}-S-1)}{2^{k+1}-k-2}$$ from Corollary \ref{C:stab} is equivalent to $$1-\frac{G}{k+1}<\stab(P).$$ In other words, the stability increases after introducing a mediator if and only if $$\stab(P) < 1-\frac{G}{k+1},$$ that is, if and only if the point representing the stability of $P$ lies to the left of the point representing $1-\frac{G}{k+1}$. The right end $G=k+1$ corresponds to the  result from Proposition \ref{P:stab}, while the case $G=0$ is trivial. We have already shown in Example \ref{E:structures1}-c) that the points $\stab(P)$ and $1-\frac{G}{k+1}$ never coincide.

We also consider the following interpretation of the previous results to be useful. Our goal here is to introduce a mediator in such a way that the stability of the political structure increases; in other words, we introduce a cone over a subcomplex so that the stability of the resulting complex is greater than that of the initial complex. We seek the smallest possible subcomplex. Only a small number of simplices may remain unaffected, and how many do so depends on the number of vertices and the stability of the initial complex. The lower the stability of the initial complex, the more simplices may remain unaffected.
\end{rem}

\end{section}

\begin{section}{Effects of introducing many mediators into a political structure}\label{S:ManyMediators}
As previously described, introducing a single mediator $v$ into a political structure $P$ yields the cone $CP=P*\{v\}$ on $P$. Introducing a mediator increases the stability of the structure, that is, $\stab(P)\leq \stab(CP)$. Introducing simultaneously two mediators $v$ and $w$ into $P$ yields the suspension $\Sigma P=P*\{v,w\}$, where $\{v,w\}$ is just the set of vertices $v$ and $w$. The authors in \cite{MV21} claim that $\stab(P)\leq  \stab(\Sigma P)$; however, this need not be true. Their argument rests on the fact that the suspension can be seen as the pushout of two copies of the cone $CP$ along $P$, together with the claim that a pushout increases stability. But the pushout does not in general increase stability. It does so only when no new vertices are introduced, i.e. when we work with $P_1 \coprod_S P_2$, where $S$ is a subcomplex of both $P_1$ and $P_2$, and $P_2$ introduces no new vertices ($V(P_2)\subseteq V(P_1)$). In the case of the suspension, of course, new vertices are necessarily added. In general, one often encounters $$\stab(P_1 \coprod_S P_2)<\stab(P_1).$$ As a simple example, consider a path $v_0v_1\ldots v_{n-1}$ on $n$ vertices $\{v_0,v_1,\ldots,v_{n-1}\}$. It is easy to see that attaching a new $1$-simplex $v_{n-1}v_n$ along the vertex $v_{n-1}$ decreases stability.



When -- following the definitions and results from \cite{MV21} -- we define the suspension as the result of introducing two mediators, we do not require these mediators to be mutually compatible. But we could also require them to be mutually compatible. This observation leads us to the following important distinction.

There are two ways to introduce $n$ mediators into a political structure:
\begin{itemize}
\item[1)] they may be introduced one by one, or – equivalently – we may require that all mediators are mutually compatible (this is the situation, for example, when a political party appears with $n$ representatives with whom everyone wishes to form a coalition),
\item[2)] they may all enter at once, simultaneously, without any mutual connection.
\end{itemize}

These two cases, in fact, represent the extreme ways of introducing $n$ mediators. In principle, one could mix the two approaches, introducing some mediators sequentially and others simultaneously; such a mixed case can then be obtained by successive application of these two extreme cases.


In the first case, introducing two mediators corresponds to the join $P*\Delta^1$, introducing three mediators corresponds to the join  $P*\Delta^2$, and so on. In general, introducing $n$ mutually compatible mediators into a political structure $P$ yields a new political structure $$P_n^1=P*\Delta^{n-1}.$$ 
In the second case, if $m_1,m_2,\ldots,m_n$ are the mediators being introduced, the resulting structure is the join $$P_n^2=P*\{m_1,m_2,\ldots,m_n\}.$$
In the case of introducing a single mediator, $P_n^1=P_n^2$ since $\Delta^0$ is a point.

\begin{example}
The structure $P_{a,b;s}$ from Example \ref{E:important} can be obtained by introducing $s$ mediators one by one into the disjoint union $$\Delta^{a-s-1}\coprod \Delta^{b-s-1}.$$
\end{example}

\begin{thm}
Let $P$ be a given political structure with $k+1$ agents, and let $P_n^1$ be the political structure obtained by introducing $n$ mediators into $P$ one by one. Then $$\lim_{n\to \infty}\stab(P_n^1)=\stab(P)+\frac{k+2}{2^{k+1}}(1-\stab(P)).$$
\end{thm}
\begin{proof}
If the number of vertices in the structure $P$ is $k+1$, then the number of vertices in $P_n^1$ is $k+1+n$.

Let $\widetilde{S}$ denote the total number of $(\geq 1)$-dimensional simplices in $P$, i.e. the total number of positive-dimensional simplices in $P$. Then $$\stab(P)=\frac{\widetilde{S}}{2^{k+1}-k-2}.$$

Each $(\geq 1)$-dimensional simplex of $P$ is also a $(\geq 1)$-dimensional simplex of $P_n^1$ , and by combining them with the simplices of $\Delta^{n-1}$, each of these generates new simplices in $P_n^1$. Taking all this into account, and based on the fact that the set of $n$ mediators has $2^n$ subsets, we conclude that $\widetilde{S}$ $(\geq 1)$-dimensional simplices in $P$ give rise to $\widetilde{S}\cdot 2^n$ $(\geq 1)$-dimensional simplices in $P_n^1$.

Furthermore, each of the $k+1$ vertices in $P$, when combined with all simplices in $\Delta^{n-1}$, of which there are $2^n-1$, contributes $$(k+1)(2^n-1)=(k+1)2^n-(k+1)$$ additional $(\geq 1)$-dimensional simplices in $P_n^1$.

Finally, the number of $(\geq 1)$-dimensional simplices in $\Delta^{n-1}$ is $2^n-n-1$.

Therefore, the total number of $(\geq 1)$-dimensional simplices in $P_n^1$ equals $$(\widetilde{S}+k+2)2^n-n-k-2$$ and hence $$\stab(P_n^1)=\frac{(\widetilde{S}+k+2)2^n-n-k-2}{2^{k+n+1}-k-n-2}.$$

By dividing both the numerator and the denominator of the previous expression by $2^n$, we simply obtain that $$\lim_{n\to \infty}\stab(P_n^1)=\frac{\widetilde{S}+k+2}{2^{k+1}},$$
which, after further transformations, becomes:
\begin{align*}
\lim_{n \to \infty} \stab(P_n^1) 
&= \frac{(\widetilde{S} + k + 2)(2^{k+1} - k - 2)}{2^{k+1}(2^{k+1} - k - 2)} \\
&= \frac{\widetilde{S} \cdot 2^{k+1} + (k+2)2^{k+1} - (k+2)^2 - \widetilde{S}(k+2)}{2^{k+1}(2^{k+1} - k - 2)} \\
&= \frac{\widetilde{S}}{2^{k+1} - k - 2} + \frac{k+2}{2^{k+1}} \left( \frac{2^{k+1} - k - 2 - \widetilde{S}}{2^{k+1} - k - 2} \right) \\
&= \frac{\widetilde{S}}{2^{k+1} - k - 2} + \frac{k+2}{2^{k+1}} \left( 1 - \frac{\widetilde{S}}{2^{k+1} - k - 2} \right) \\
&= \stab(P) + \frac{k+2}{2^{k+1}} \big(1 - \stab(P) \big).
\end{align*}
\end{proof}

The number $\frac{k+2}{2^{k+1}}$ is small for large $k$. 
Therefore, introducing many mediators does not significantly increase the stability of the structure, especially in the case of structures with a larger number of agents. In this context, one may pose the following question, which would require a more formal framework to answer: given that introducing mediators brings with it a cost while increasing stability yields a benefit, how many mediators should be introduced for their introduction to be justified?
\begin{thm}
Let $P$ be a political structure with $k+1$ agents and $\widetilde{S}$ positive-dimensional simplices, and let $P_n^2$ denote the political structure obtained by introducing $n$ mediators into $P$ simultaneously. Then
$$\stab(P_n^2)>\stab(P)$$ if and only if $$\widetilde{S}<\frac{n(k+1)(2^{k+1}-k-2)}{2^{k+1}(2^n-n-1)+n(k+1)}.$$
\end{thm}
\begin{proof}
If $P$ has $\widetilde{S}$ $(\geq 1)$-dimensional simplices, then it is easy to conclude that the total number of $(\geq 1)$-dimensional simplices in $P_n^2$ is $$\widetilde{S}+n\widetilde{S}+n(k+1)=(n+1)\widetilde{S}+n(k+1).$$ The number of vertices in $P_n^2$ is $k+1+n$. Therefore, 
$$\stab(P_n^2)=\frac{(n+1)\widetilde{S}+n(k+1)}{2^{k+n+1}-n-k-2}.$$
After elementary calculations, we obtain that the inequality $\stab(P_n^2)>\stab(P)$ is equivalent to $$\widetilde{S}<\frac{n(k+1)(2^{k+1}-k-2)}{2^{k+1}(2^n-n-1)+n(k+1)}.$$
\end{proof}
Therefore, introducing $n$ mediators simultaneously will increase stability only if $P$ has a small number of positive-dimensional simplices.
\begin{cor}
Let $P$ be a political structure with $k+1$ agents and $\widetilde{S}$ positive-dimensional simplices, and let $\Sigma P$ denote the suspension of $P$. Then
$$\stab(\Sigma P)>\stab(P)$$ if and only if $$\widetilde{S}<\frac{2(k+1)(2^{k+1}-k-2)}{2^{k+1}+2(k+1)}.$$
\end{cor}
\begin{proof}
This follows directly from $P_2^2=\Sigma P$.
\end{proof}

Let us also consider the $n$-fold suspension $\Sigma^n P$ of a political structure $P$. In this construction, pairs of new mediators are introduced $n$ times: first, two mediators are added, yielding the suspension $\Sigma P$; then, two mediators are added to $\Sigma P$, yielding $\Sigma^2 P$, and so on, up to $\Sigma^n P$.

\begin{thm}
Let $P$ be a political structure with $k+1$ agents and $\Sigma^n P$ the $n$-fold suspension of $P$. Then 
$$\lim_{n \to \infty} \stab(\Sigma^n P)=0.$$
\end{thm}
\begin{proof}
If $P$ has $k+1$ vertices and $\widetilde{S}$ $(\geq 1)$-dimensional simplices, then the number of vertices in $\Sigma P$ is $k+3$, and the number of $(\geq 1)$-dimensional simplices in $\Sigma P$ is $$\widetilde{S}+2\widetilde{S}+2(k+1)=3\widetilde{S}+2(k+1).$$ The number of vertices in $\Sigma^2 P$ is $k+5$, and the number of $(\geq 1)$-dimensional simplices in $\Sigma^2 P$ is $$3(3\widetilde{S}+2(k+1))+2(k+3)=3^2 \widetilde{S}+3\cdot 2(k+1)+2(k+3).$$ By induction, it follows that the number of vertices in $\Sigma^n P$ is $k+2n+1$, and the number of $(\geq 1)$-dimensional simplices in $\Sigma^n P$ is $$3^n \widetilde{S}+3^{n-1}\cdot 2(k+1)+3^{n-2}\cdot 2(k+3)+\cdots + 3\cdot 2(k+2n-3)+2(k+2n-1).$$
Therefore, $$\stab(\Sigma^n P)=\frac{3^n \widetilde{S}+3^{n-1}\cdot 2(k+1)+3^{n-2}\cdot 2(k+3)+\cdots + 3\cdot 2(k+2n-3)+2(k+2n-1)}{2^{2n+k+1}-k-2n-2}.$$
By dividing both the numerator and the denominator of the previous equation by $2^{2n}$, we obtain $$\lim_{n \to \infty} \stab(\Sigma^n P)=0.$$

\end{proof}

\end{section}

\begin{section}{Splitting and stability}\label{S:Splitting}
Let $P$ be a political structure consisting of two disjoint connected components, $P_1$ and $P_2$. We can consider the situation in which one agent from $P_1$ and one agent from $P_2$ choose to coordinate their actions, thereby effectively functioning as a single agent -- in other words, the agents merge. From a topological perspective, the disjoint union $P_1 \amalg P_2$ thereby becomes the wedge sum $P_1 \lor P_2$. As shown in \cite[Proposition 3.13]{MV21}, this transformation results in an increase in the stability, that is, $\stab(P_1 \amalg P_2)<\stab(P_1 \lor P_2)$.

In contrast to agent merging, one can also theoretically and practically consider the reverse process -- agent splitting -- which is frequently exemplified in real-world cases such as the division of a political party into two or more separate parties. We will investigate how the stability of the structure changes in this case; more precisely, we will analyze the relationship between the stability of the political structure before agent splitting and that of the newly formed structure after the splitting.

Assume that $v$ is an agent (a vertex) of $P$ and that $v$ splits into two agents $a$ and $b$ (new vertices). Each of these new agents forms new coalitions, which can be described as the \textit{links} $L_a \subset P$ and $L_b \subset P$ of the vertices $a$ and $b$, respectively. Therefore, the resulting political structure is $$P_v'=(P\setminus \{v\})\amalg \left(a*L_a\right) \amalg \left(b*L_b\right).$$ The question is to compare the stability of $P$ and $P_v'$.

\begin{prop}\label{P:split}
Let $P$ be a political structure with $k+1$ vertices and the total number $S$ of simplices.  Suppose a vertex $v$ splits into two vertices $a$ and $b$, resulting in the political structure $P_v'$. Let $F$ denote the number of coalitions containing $v$, that is, the number of simplices in which the vertex $v$ appears, not counting the vertex $v$ itself. Let $X$ denote the total number of new coalitions, i.e., new simplices containing $a$ or $b$, after the split of $v$, excluding the vertices $a$ and $b$ themselves. Then $$\stab(P)>\stab(P_v')$$ if and only if 
\begin{equation}\label{eq:XF}
X<\frac{(2^{k+1}-1)(S-k-1)}{2^{k+1}-k-2}+F.
\end{equation}
\end{prop}
\begin{proof}
If $P$ has $k+1$ vertices and $S$ simplices, then the structure $P_v'$ has $k+2$ vertices and the total number $$S-1-F+2+X=S+X-F+1$$ of simplices. The stabilities of these two structures are then $$\stab(P)=\frac{S-(k+1)}{2^{k+1}-(k+2)}$$ and $$\stab(P_v')=\frac{S+X-F+1-(k+2)}{2^{k+2}-(k+3)}=\frac{S+X-F-(k+1)}{2^{k+2}-(k+3)}.$$
After elementary calculations, we obtain that the inequality $\stab(P)>\stab(P_v')$ is equivalent to the inequality $$X-F<\frac{(2^{k+1}-1)(S-k-1)}{2^{k+1}-k-2}.$$
\end{proof}
In other words, the stability after splitting increases, i.e., $\stab(P)<\stab(P_v')$, if and only if $$X>\frac{(2^{k+1}-1)(S-k-1)}{2^{k+1}-k-2}+F.$$

We see that $X$ can be regarded as a characteristic of the simplicial complex and the vertex being split.

\begin{rem}
Notice that $F = (2^k-1)\via(v)$. Hence, the inequality \eqref{eq:XF} can be written as $$X < (2^{k+1}-1)\stab(P)+(2^k-1)\via(v).$$ Therefore, $X$ is a characteristic described using the notions introduced in \cite{MV21}. In addition, in this case one could proceed similarly to the discussion in Remark \ref{R:G}.
\end{rem}

Proposition \ref{P:split} confirms the expected fact: the greater the total number of compatibilities in the political structure before the split of agent $v$ into agents $a$ and $b$, and the greater the number of compatibilities between $v$ and the other agents, the more new compatibilities $a$ and $b$ need to establish in order for the stability of the new political structure to exceed that of the original one. In short, stability increases if the number of new coalitions is sufficiently large, and decreases if this number is sufficiently small.

\begin{example}\label{E:split}
Consider the political structure from Example \ref{E:structures1}-b). Its stability is $3/13$. Suppose that the vertex $v_3$ splits into two vertices, 
$a$ and $b$, where $a$ is compatible with $v_0, v_1, v_2$ and with all of their possible mutual coalitions, and $b$ is compatible only with $v_4$. The resulting structure 
\begin{equation*}
\begin{aligned}
P_{v_3}' = \{ &\{v_0,v_1,v_2,a\}, \{v_0,v_1,v_2\},\{v_0,v_1,a\}, \{v_0,v_2,a\}, \{v_1,v_2,a\}, \{v_0,v_1\}, \{v_1,a\}, \{v_1,v_2\}, \\
             &\{v_0,a\}, \{a,v_2\}, \{v_0,v_2\}, \{v_1,v_4\}, \{b,v_4\}, \{v_0\}, \{v_1\}, \{v_2\}, \{a\}, \{b\}, \{v_4\} \}
\end{aligned}
\end{equation*}
 is depicted in Figure \ref{fig:split}(B).

It may seem at first glance that the new structure has greater stability than the original, given that the vertex $a$ contributes to the formation of four $2$-simplices and one $3$-simplex. However, this is not the case: $\stab(P_{v_3}')=13/57<\stab(P)$. In this case, $X=8$, $F=1$, $S=11$, $k+1=5$, and $$8=X<\frac{(2^{k+1}-1)(S-k-1)}{2^{k+1}-k-2}+F\approx 8.15.$$
\end{example}

\begin{figure}[h!]
\centering

\begin{subfigure}[t]{0.45\textwidth}
\centering
\begin{tikzpicture}[scale=0.8]
\fill[blue!20] (0,0) -- (4,0) -- (2,3.465) -- cycle;
\draw (0,0) -- (4,0) -- (2,3.465) -- (0,0);
\draw (4,0) -- (6,3.465);
\draw (4,0) -- (8,0);
\draw[fill] (0,0) circle [radius=0.1] node [left] {$v_0$};
\draw[fill] (4,0) circle [radius=0.1] node [below] {$v_1$};
\draw[fill] (2,3.465) circle [radius=0.1] node [above] {$v_2$};
\draw[fill] (6,3.465) circle [radius=0.1] node [right] {$v_3$};
\draw[fill] (8,0) circle [radius=0.1] node [right] {$v_4$};
\end{tikzpicture}
\caption{Political structure $P$}
\end{subfigure}
\hfill
\begin{subfigure}[t]{0.45\textwidth}
\centering
\begin{tikzpicture}[scale=0.8]

\coordinate (A) at (4.8,2.3);

\fill[blue!15, opacity=0.7] (0,0) -- (4,0) -- (2,3.465) -- cycle;     
\fill[blue!30, opacity=0.7] (A) -- (4,0) -- (2,3.465) -- cycle;       
\fill[blue!10, opacity=0.7] (0,0) -- (4,0) -- (A) -- cycle;           
\fill[blue!20, opacity=0.6] (0,0) -- (2,3.465) -- (A) -- cycle;       

\draw (0,0) -- (4,0) -- (2,3.465) -- (0,0);
\draw (A) -- (4,0);
\draw (A) -- (2,3.465);
\draw[dashed] (A) -- (0,0);

\draw (4,0) -- (8,0);
\draw (8,0) -- (6.5,3.465);

\draw[fill] (0,0) circle [radius=0.1] node [left] {$v_0$};
\draw[fill] (4,0) circle [radius=0.1] node [below] {$v_1$};
\draw[fill] (2,3.465) circle [radius=0.1] node [above] {$v_2$};
\draw[fill] (A) circle [radius=0.1] node [above right] {$a$};
\draw[fill] (8,0) circle [radius=0.1] node [right] {$v_4$};
\draw[fill] (6.5,3.465) circle [radius=0.1] node [right] {$b$};
\end{tikzpicture}
\caption{Political structure $P_{v_3}'$}
\end{subfigure}

\caption{Political structures from Example \ref{E:split}}
\label{fig:split}
\end{figure}
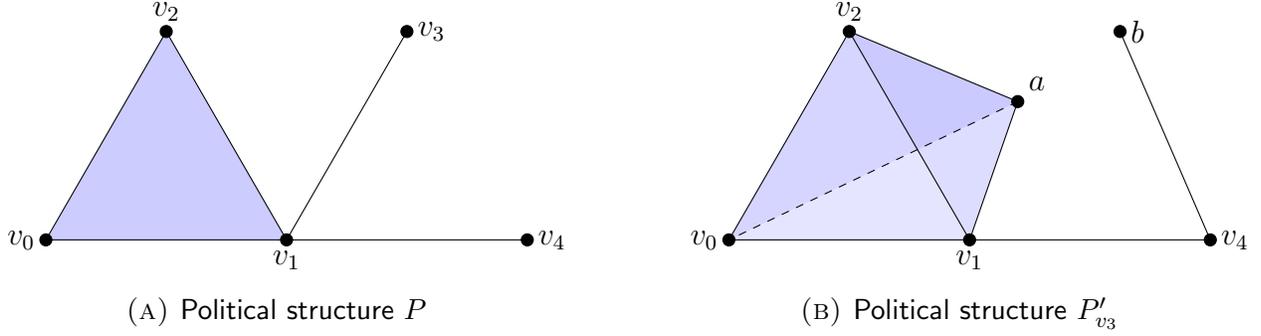

\begin{example}
The structure $P_{a,b;s}$ from Example \ref{E:important} can be obtained from $P_{a,b;s+1}$ by splitting one of the common agents in such a way that one of the new agents is compatible with all faces of $\Delta^{a-1}$, while the other is compatible with all faces of $\Delta^{b-1}$. Since the function $\stab(P_{a,b;s})$ is increasing with respect to $s$, we have
$$\stab(P_{a,b;s})<\stab(P_{a,b;s+1}).$$

\end{example}

\end{section}

\begin{section}{Shellable political structures and their stability}\label{S:Shelling}


The structure $P_{a,b;s}$ from Example \ref{E:important}, and in particular the fact that the maximal stability of $P_{b,b;s}$ is attained when $s=b-1$, is one of the motivations for introducing the concept of \textit{shellability} in our context.

In this section, we examine a special class of political structures modeled by simplicial complexes -- those that are shellable. Shellability is a combinatorial property that ensures a well-behaved structure of facets and yields non-negative, interpretable entries in the $h$-vector. By leveraging this property, we obtain clearer insight into the stability of such structures and how it evolves under controlled modifications. We also consider weighted models and, in the last section, graph-based interpretations within the same combinatorial framework.

We begin by expressing stability in terms of the $h$-vector, proceed to discuss the notion of shellability, and ultimately investigate the relevance of shellability in the context of stability and weighted political structures.

The stability of a political
structure $P$ depends solely on its $f$-vector $f(P)=(f_{-1},f_0,\ldots,f_d)$, more precisely, on
the number of faces of $P$ of dimension at least $1$. The
$h$-vector of a $d$-dimensional simplicial complex $P$ is a
$(d+2)$-tuple $h(P)=(h_0,h_1,\ldots,h_{d+1})$ defined by
\begin{equation}\label{E:h-iz-f}
h_s=\sum_{i=0}^{s}(-1)^{s-i}{d+1-i\choose d+1-s}f_{i-1}, \textrm{
for all }s=0,1,\ldots,d+1.
\end{equation}
In other words, $h(P)$ is a linear
transformation of $f(P)$.
 Note that for any $d$-dimensional complex $P$ we have $h_0=1$, $h_1=f_0-d-1$.
 Moreover, the last entry $h_{d+1}$ of the $h$-vector equals
 the reduced Euler characteristic of $P$:
 $$h_{d+1}=f_d-f_{d-1}+\cdots+(-1)^df_0+(-1)^{d+1}=\widetilde{\chi}(P).$$
 
 \begin{example}
 The $f$-vector of the structure in Example \ref{E:structure1} is $(1,6,6,1)$, and the corresponding $h$-vector is $(1,3,-3,0)$.
 
 The $f$-vector of the structure in Example \ref{E:structure2} is $(1,5,7,4,1)$, and the corresponding $h$-vector is $(1,1,-2,1,0)$.
 \end{example}

 The linear transformation (\ref{E:h-iz-f}) is invertible; indeed,
 the inverse is given by
\begin{equation}\label{E:f iz h}
f_{i-1}=\sum_{j=0}^i{d+1-j\choose  i -j} h_j\textrm{, for all
}i=0,1,\ldots,d+1.
\end{equation}
The stability of a political structure $P$ can thus be expressed
in terms of its $h$-vector.
\begin{prop}\label{P:stab-via-h}
Let $P$ be a $d$-dimensional political structure with $k+1$ agents, and let its $h$-vector be $(h_0,h_1,\ldots,h_{d+1})$. Then
\begin{equation}\label{E:Stabviah}
    \stab(P)=
\frac{\sum_{i=0}^{d+1} 2^{d+1-i}h_i-k-2}{2^{k+1}-k-2}.
\end{equation}
\end{prop}
\begin{proof}
Using \eqref{E:f iz h}, the facts $h_0=1$, $f_0=k+1$, and elementary operations, we obtain \begin{equation*}
\begin{aligned}
f_1 + \cdots + f_d 
&= \sum_{i=1}^{d+1} f_{i-1} - f_0 \\
&= \sum_{i=1}^{d+1} \sum_{j=0}^i \binom{d+1 - j}{i - j} h_j - f_0 \\
&= \sum_{i=0}^{d+1} 2^{d+1 - i} h_i - h_0 - f_0 \\
&= \sum_{i=0}^{d+1} 2^{d+1 - i} h_i - k - 2
\end{aligned}
\end{equation*}
 from which \eqref{E:Stabviah} follows.
\end{proof}



The combinatorial meaning of the entries of the $h$-vector, as given by the definition in \eqref{E:h-iz-f}, is not immediately evident and does not always exist, in particular for negative entries. However, for a well-behaved class of complexes known as \emph{shellable complexes}, the entries of $h$-vector are non-negative and admit a meaningful combinatorial interpretation. Also, some of the operations that appear here (cone, suspension, join) preserve shellability and the change of $h$-vector under these operations is easy to track. We briefly recall the definition of shellability; for further details, see \cite{BjornerH, BjornerShellable, BjornerWachs, Ziegler}.

We recall that a \textit{face} of a simplicial complex is any simplex in the complex. A maximal simplex in a simplicial complex -- one that is not a face of any other simplex in the complex -- is called a \textit{facet}. A $d$-dimensional simplicial complex $P$ is said to be
\textit{pure} if every simplex of $P$ of dimension less than $d$ is a
face of some $d$-dimensional simplex in $P$. A simplicial complex $P$ is
\textit{shellable} if $P$ is pure and there exists a linear
ordering (called a \textit{shelling order}) $F_1, F_2,\ldots, F_k$ of facets of $P$ such that for every $i <
j\leq k$ there exists some $l < j$ and a vertex $v$ of $F_j$ such
that
$$F_i\cap F_j \subseteq F_l\cap F_j =F_j\setminus\{v\}.$$


Informally speaking, a shellable $d$-dimensional simplicial complex $P$ can be constructed by successively gluing its facets in such a way that each new facet $F_j$ attaches to the union of the previously added facets $F_1\cup
F_2\cup\cdots\cup F_{j-1}$ along a pure $(d-1)$-dimensional subcomplex  (i.e., a union of facets of $F_j$).

As an example, the boundary of every simplicial polytope is shellable. From a topological perspective, a shellable simplicial complex of dimension $d$ is homotopy equivalent to a wedge of $d$-dimensional spheres or is contractible.

\begin{rem}
In summary, a shellable political structure $P$ of dimension $d$ satisfies the following two conditions:
\begin{itemize}
\item all maximal viable configurations of $P$ have the same number of agents, namely $d+1$;
\item for any two maximal viable configurations $F$ and $F'$ of $P$, there exists a path $F = F_1, F_2, \ldots, F_n = F'$ such that $|F_i \cap F_{i+1}| = d$ for each $i$.
\end{itemize}
This can be interpreted as follows. 
Any maximal viable configuration (coalition) can be obtained from any other through small changes, where one agent is replaced by another. At each step, the set remains a maximal viable configuration, allowing one to move from one coalition to another via small, incremental steps without drastic disruptions.
\end{rem}

  For a fixed shelling order $F_1, F_2,\ldots, F_k$ of
   $P$, we define the \textit{restriction}
   $R(F_j)$ of the facet $F_j$ by:
$$R(F_j) = \{v \textrm{ is a vertex of }F_j :
F_j \setminus \{v\}\subset F_i\textrm{ for some }1 \leq i < j\}.$$
In other words, if we built up $P$ according to the shelling
order, then $R(F_j)$ is a minimal new face at the $j$-th step. The
\textit{type} of the facet $F_j$ in the given shelling order is
the cardinality of $R(F_j)$, i.e., $\type(F_j) = |R(F_j) |$.

If a simplicial complex $P$ is shellable, then
$$h_k(P)=|\{F\textrm{ is a facet of }P: \type(F)=k\}|$$ is
an important combinatorial interpretation of $h(P)$. This
interpretation of the $h$-vector played a central role in the
proof of the Upper bound theorem and in the characterization of
$f$-vectors of simplicial polytopes; see \cite{Ziegler} for details.

\begin{example}
The $h$-vector of the simplicial complex $P_{b,b-1;b-2}$, which is not shellable, is $(1,1,-1,0,\ldots,0)$.

The simplicial complex $P_{b,b;b-1}$ is shellable. Its $h$-vector is $(1,1,0,\ldots,0)$.
\end{example}

If the political structure under consideration is shellable, we observe that formula \eqref{E:Stabviah} becomes particularly useful.
In that case, we can easily track how the stability
changes by adding or deleting some facets (keeping the
shellability). Adding (or removing) a new facet of type $i>1$
(there are no new vertices, i.e. agents), increases (or
decreases) the corresponding entry $h_i$ by one.
Consequently, the stability of $P$ increases (or decreases) by
$$\frac{2^{d+1-i}}{2^{k+1}-k-2}.$$

\begin{example}\label{E:shelling-decrease}
Consider what happens if we glue a facet that brings one new vertex, i.e., a facet of type $1$. Assume that $$P=F_1\cup F_2 \cup \cdots \cup F_j$$ and let $F_{j+1}$ be a $d$-dimensional simplex of type 1, where $P$ has $k+1$ agents. If $\widetilde{S}$ is the total number of positive-dimensional simplices in $P$, then $\stab(P)=\frac{\widetilde{S}}{2^{k+1}-k-2}$. If the dimension of $P$ is $d$, then for $$P'=F_1\cup F_2 \cup \cdots \cup F_j\cup F_{j+1}$$ we obtain 
$$\stab(P')=\frac{\widetilde{S}+2^d-1}{2^{k+2}-k-3}.$$
Then the inequality $$\stab(P')>\stab(P)$$ is equivalent to the inequality $$\widetilde{S}<\frac{(2^d-1)(2^{k+1}-k-2)}{2^{k+1}-1}.$$ Since $2^{k+1}-k-2<2^{k+1}-1$, it follows that $\stab(P')>\stab(P)$ implies $\widetilde{S}<2^d-1$, which is impossible since $P$ is $d$-dimensional. Therefore, in this case, the inequality $$\stab(P')<\stab(P)$$ holds; that is, the stability decreases.
\end{example}


The above reasoning becomes particularly useful when considering a \textit{weighted simplicial complex model}. We briefly recall the definition introduced in \cite{MV21}, presented here in a modified form. Let $I=\{p_1,p_2,\ldots,p_m\}$ be a set of $m$ issues that a group of $k+1$ agents $A=\{a_0,a_1,\ldots,a_k\}$ must evaluate. For each issue $p_j\in I$, we associate a simplicial complex $P_j$, whose faces represent viable configurations for that issue. The resulting \textit{weighted political structure} is then given by the tuple $\mathcal{P}=(P_1,P_2,\ldots,P_m)$.

 The weighted stability
 of a weighted political structure $\mathcal{P}$ is just 
the arithmetic mean of the individual stabilities of the complexes
 $P_1,P_2,\ldots,P_m$, i.e.
 $$\wstab(\mathcal{P})=\frac{1}{m}\Big(\stab(P_1)+\stab(P_2)+\cdots+\stab(P_m)\Big).$$
 
The definition of the weighted version given in \cite{MV21} is sufficient for the purposes there, as it only requires the total number of simplices. Our definition, however, is more informative, as it keeps precise track of which agents are aligned on which question. Moreover, if the complex is shellable, the relevant quantities can be easily computed under any changes in preferences or alignments. In other words, if the simplicial complex changes -- either from question to question or over time -- while remaining shellable, the change in stability can be tracked and the weighted stability can be computed.

If all complexes $P_j$ are shellable, and if $P_{j+1}$ is obtained
from the previous one by adding a few new faces, we can easily
calculate the stabilities and the weighted stability. This in itself is a reasonable motivation to consider shellable structures in this context.

\begin{example}[The boundary of a simplex] Any linear ordering of the
facets of a $k$-dimensional simplex is a shelling order, and the type of the
$i$-th facet in this order is $i-1$. Therefore, we have that
$h_i(\partial \Delta^k)=1$ for all $i=0,1,\ldots,k$. If
$P=\partial\Delta^k$, then the stability of $P$ attains the
maximal nontrivial value among all political structures with $k+1$
vertices
$$\stab(P)=\frac{\sum_{i=0}^{d+1} 2^{d+1-i}h_i-k-2}{2^{k+1}-k-2}=\frac{\sum_{i=0}^{k} 2^{k-i}-k-2}{2^{k+1}-k-2}=\frac{2^{k+1}-k-3}{2^{k+1}-k-2}.$$

Assume that we consider a set $I=\{p_1,p_2,\ldots,p_m\}$ of $m\leq k+1$ issues, and that a subcomplex $P_j$ of $\partial\Delta^k$ is the union of exactly
$j$ facets of $\Delta^k$. Then exactly first $j$ entries $h_0,h_1,\ldots, h_{j-1}$ of $h(P_j)$ equal $1$, while all others are $0$. From  (\ref{E:Stabviah}) we obtain that, for $j>2$,
$$\stab(P_{j})=\frac{2^k+2^{k-1}+\cdots+2^{k-j+1}-k-2}{2^{k+1}-k-2}=
\stab(P_{j-1})+\frac{2^{k-j+1}}{2^{k+1}-k-2}.$$

For $j=2$ this recurrence formula does not hold, since here $P_1$ is the union of a $(k-1)$-simplex and a point, hence not shellable, and its stability equals $$\frac{2^k-k-1}{2^{k+1}-k-2}.$$

If we consider the
set of issues $I=\{p_1,p_2,\ldots,p_m\}$ for some $m\leq k+1$, the
weighted stability of the corresponding political structure $\mathcal{P}=(P_1,P_2,\ldots,P_m)$ is then given by
$$\wstab(\mathcal{P})=\frac{2^k+\frac{m-1}{m}2^{k-1}+\cdots+
\frac{1}{m}2^{k-m+1}+\frac{1}{m}-k-2}{2^{k+1}-k-2}.$$
\end{example}
\end{section}

\begin{section}{The independence complex of a graph as a political structure}
Here we describe a natural way to model political
structures using graphs and provide two illustrative examples.

Consider a graph $G=(V,E)$ with $k+1$
vertices. Assume that the vertices of $G$ represent the agents.
Let any two agents connected by an edge in $G$ be in conflict
and therefore unable to appear together in any viable configuration. In
this setting, the set of all viable configurations corresponds
precisely to the independence complex $I(G)$ of the graph $G$, which is the simplicial complex whose faces are the independent sets of vertices of $G$.

For a formal definition and more on independence complexes, including their connection to shellability, see \cite{Engstrom}.
\begin{example}
Suppose $k+1 = n$ agents are arranged along a path, and each is in
conflict with its immediate neighbors. Then the viable
configurations correspond to independent sets in the path graph
$P_n$. The $f$-vector entries of this political structure $P =
I(P_n)$ count the subsets of $[n]=\{1,\ldots,n\}$ with no consecutive elements.
In particular,
$$f_{i-1}(P)={n-i+1\choose i}.$$
 Using the well-known fact that the total
number of all faces of $I(P_n)$ is $(n+2)$-th Fibonacci number
$F_{n+2}$, we obtain that
$$\stab(P)=\frac{F_{n+2}-n-1}{2^n-n-1}.$$
\end{example}
\begin{example}
Consider $k+1 = 2n$ agents labeled
 by the set $\{1, 2, \ldots, n\} \cup \{-1, -2, \ldots, -n\}$.
 Let each agent $i$ be in conflict solely with agent $-i$. These conflicts can then be represented by a graph consisting of $n$ edges, each connecting $i$ and $-i$ for all $i \in [n]$.
The corresponding political structure is $P=
\partial C^{\triangle}_{n-1}$, the boundary of $(n-1)$-dimensional
crosspolytope $C^{\triangle}_{n-1}$ (the dual of $(n-1)$-cube).
The complex $\partial C^{\triangle}_{n-1}$ has $2^n$ facets, each
corresponding to an ordered partition of $[n]$ into two subsets.
 For a partition $[n]=A\sqcup B$, the facet associated to $(A,B)$
 is defined by $$F_{(A,B)}=\{i:i\in A\}\cup \{-j:j\in B\}.$$
There exists a natural shelling of $\partial C^{\triangle}_{n-1}$,
and we describe it here for the sake of completeness:
\begin{equation}\label{E:Shellocta}
F_{(A,B)}<F_{(A',B')}\Leftrightarrow \left\{%
\begin{array}{ll}
    |B|<|B'|,  \textrm{ or } \\
 |B|=|B'| \textrm{ and }\min\, (B\triangle B') \in B\\
\end{array}%
\right..
\end{equation}

 In this order the type of $F_{(A,B)}$ is simply $|B|$.
 The entries of $h(\partial C^{\triangle}_{n-1})$ thus are given by
$h_i={n \choose i}$. Using equation (\ref{E:Stabviah}), we compute
the stability of $P$ as
\begin{equation*}
\begin{aligned}
\stab(P) 
&= \frac{\displaystyle \sum_{i=0}^{d+1} 2^{d+1 - i} h_i - k - 2}{2^{k+1} - k - 2} \\[10pt]
&= \frac{\displaystyle \sum_{i=0}^{n} 2^{n - i} \binom{n}{i} - (2n - 1) - 2}{2^{2n} - (2n - 1) - 2} \\[10pt]
&= \frac{3^n - 2n - 1}{2^{2n} - 2n - 1}.
\end{aligned}
\end{equation*}
Note that under the shelling order~(\ref{E:Shellocta}) all facets
of type $i$ appear before those of type $j$ for any $0 \leq i <
j$.

We now analyze the following weighted model.  
Consider an agenda $I = \{p_1, p_2, \ldots, p_n\}$ consisting of $n$ issues.  
Assume that $P_j$, the set of viable configurations for the issue $p_j$, for all $1\leq j\leq n$, is the subcomplex of $\partial C^{\triangle}_{n-1}$ spanned by all facets of type at most $j$. 
The facets of $P_j$ form an initial segment of the shelling described in \eqref{E:Shellocta}.
From the formula \eqref{E:Stabviah}, it follows that  
$$\stab(P_j)=\frac{2^n+n2^{n-1}+{n\choose 2}2^{n-2}+\cdots+
{n\choose j}2^{n-j}-2n-1}{2^{2n}-2n-1}.$$ 

Following the reasoning
described after Example \ref{E:shelling-decrease}, and after straightforward calculations, we obtain that the weighted
stability of $\mathcal{P}=(P_1,P_2,\ldots,P_n)$ is
$$\wstab (\mathcal{P})=\stab(P_1)+\sum_{i=2}^n\frac{(n+1-i)
{n\choose i}2^{n-i}}{2^{2n}-2n-1}.$$
\end{example}

\end{section}

\begin{section}{Further questions}\label{S:Future}
\begin{enumerate}
\item[1.] According to the definition of stability of a political structure, two structures with the same number of vertices and the same total number of $(\geq 1)$-dimensional simplices will have the same stability. However, these structures need not be topologically equivalent. Therefore, the stability defined in this manner is not a topological invariant, which makes it problematic if one aims to approach the problem from a more topological perspective.

Structures with $k+1$ vertices are partitioned into $2^{k+1}$ ``stability classes". Even the dimension of the structure does not appear as something relevant.

Can we define a natural refinement of $\stab(P)$ that distinguishes certain political structures with the same number of vertices and the same number of positively-dimensional simplices? Perhaps we can incorporate viability into it, which has been mentioned in passing in \cite{MV21}.  For example, consider the following three structures with five vertices (Figure \ref{F:via}):
\begin{itemize}
\item $P_1$, which consists of two triangles sharing a common vertex;
\item $P_2$, which consists of two triangles sharing a common edge, along with an additional edge connecting one of the vertices not on the common edge to an extra vertex;
\item $P_3$, which consists of two triangles sharing a common edge, with an additional edge connecting one of the vertices on the common edge to an extra vertex.
\end{itemize}

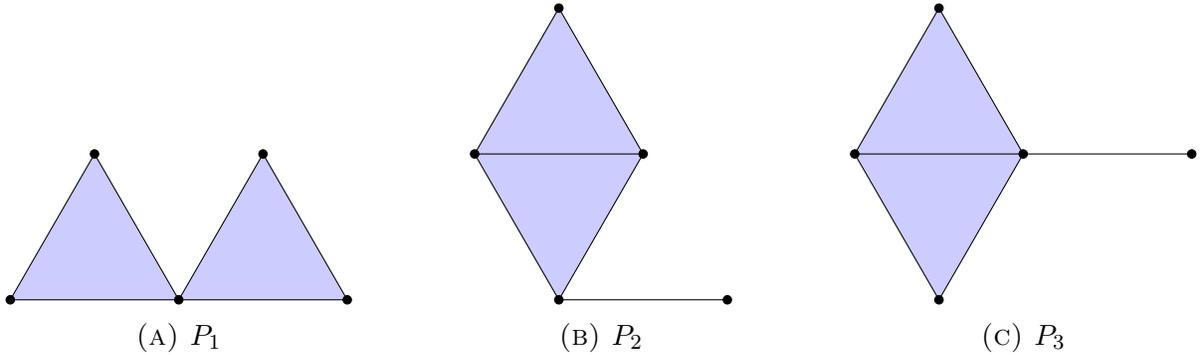
\begin{figure}[h!]
\centering

\begin{subfigure}[t]{0.32\textwidth}
\centering
\begin{tikzpicture}[scale=0.56]
\fill[blue!20] (0,0) -- (4,0) -- (2,3.465) -- cycle;
\fill[blue!20] (8,0) -- (4,0) -- (6,3.465) -- cycle;
\draw (0,0) -- (4,0) -- (2,3.465) -- (0,0);
\draw (4,0) -- (6,3.465);
\draw (4,0) -- (8,0);
\draw (8,0) -- (6,3.465);
\draw[fill] (0,0) circle [radius=0.1];
\draw[fill] (4,0) circle [radius=0.1];
\draw[fill] (2,3.465) circle [radius=0.1];
\draw[fill] (6,3.465) circle [radius=0.1];
\draw[fill] (8,0) circle [radius=0.1];
\end{tikzpicture}
\caption{$P_1$}
\end{subfigure}
\hfill
\begin{subfigure}[t]{0.32\textwidth}
\centering
\begin{tikzpicture}[scale=0.56]
\fill[blue!20] (0,0) -- (4,0) -- (2,3.465) -- cycle;
\fill[blue!20] (0,0) -- (4,0) -- (2,-3.465) -- cycle;
\draw (0,0) -- (4,0) -- (2,3.465) -- (0,0);
\draw (4,0) -- (2,-3.465);
\draw (0,0) -- (2,-3.465);
\draw (6,-3.465) -- (2,-3.465);
\draw[fill] (0,0) circle [radius=0.1];
\draw[fill] (4,0) circle [radius=0.1];
\draw[fill] (2,3.465) circle [radius=0.1];
\draw[fill] (2,-3.465) circle [radius=0.1];
\draw[fill] (6,-3.465) circle [radius=0.1];
\end{tikzpicture}
\caption{$P_2$}
\end{subfigure}
\hfill
\begin{subfigure}[t]{0.32\textwidth}
\centering
\begin{tikzpicture}[scale=0.56]
\fill[blue!20] (0,0) -- (4,0) -- (2,3.465) -- cycle;
\fill[blue!20] (0,0) -- (4,0) -- (2,-3.465) -- cycle;
\draw (0,0) -- (4,0) -- (2,3.465) -- (0,0);
\draw (4,0) -- (2,-3.465);
\draw (0,0) -- (2,-3.465);
\draw (8,0) -- (4,0);
\draw[fill] (0,0) circle [radius=0.1];
\draw[fill] (4,0) circle [radius=0.1];
\draw[fill] (2,3.465) circle [radius=0.1];
\draw[fill] (2,-3.465) circle [radius=0.1];
\draw[fill] (8,0) circle [radius=0.1];
\end{tikzpicture}
\caption{$P_3$}
\end{subfigure}

\caption{Political structures with the same stabilities but different viabilities}
\label{F:via}
\end{figure}

We have $$\stab(P_1)=\stab(P_2)=\stab(P_3).$$

However, the highest viability in the set of viabilities of all agents in $P_1$ is $6/15$, the same as the highest viability in the set of all agents in $P_3$, while the highest viability in the set of viabilities of all agents in $P_2$ is $5/15$. The next highest viability in  $P_1$ is $4/15$, and in $P_3$ it is $6/15$. Therefore, we might say that $\widetilde{\stab}(P_3)>\widetilde{\stab}(P_1)>\widetilde{\stab}(P_2)$
for some refined function $\widetilde{\stab}(P)$. This refined notion of stability can serve as an improvement over the basic definition of stability. 

Another idea to distinguish structures, at least by their $f$-vectors, is to define
$$\widehat{\stab}(P)=\frac{1}{f_0}\sum_{i=0}^{\dim{P}} \frac{f_i}{\binom{f_0}{i+1}}.$$
With this definition, for a fixed total number of simplices, the function $\widehat{\stab}$ favors structures that have more simplices of lower dimension, in the sense that they attain a higher value of $\widehat{\stab}(P)$.



\item[2.] The fact that introducing a mediator increases the stability of a political structure -- that is, $\stab(P)\leq \stab(P*\{v\})$ -- motivates the following

\begin{conj}
Let $P_1$ and $P_2$ be political structures such that $\stab(P_1)<\stab(P_2)$. Then $$\stab(P_1)<\stab(P_1*P_2)<\stab(P_2).$$
\end{conj}

\item[3.] The homology of a political structure can often be easily determined from shellability. 

When a political structure $P$ is shellable, introducing one or more mediators into $P$ preserves shellability and the changes in the $h$-vector -- which remain unchanged when a single mediator is added (that is, if $P$ is shellable, then the cone $CP$ also remains shellable and retains the same $h$-vector) -- can be easily tracked. On the other hand, introducing a mediator into a substructure of $P$ is a more subtle situation regarding shellability and deserves further investigation.

There are also more advanced techniques for establishing shellability and for computing $h$-vectors, such as lexicographic shellings for posets and non-pure shellings. A more systematic use of these methods, and their possible adaptation to the political framework considered here, could be explored in future work.

\end{enumerate}

\end{section}

\bibliographystyle{alpha}

\pagestyle{plain}

\end{document}